\documentclass[letterpaper, 10 pt, conference]{ieeeconf}  

\IEEEoverridecommandlockouts                              

\usepackage{cite}
\usepackage{amsmath,amssymb,amsfonts}
\usepackage{tabularx,booktabs}
\usepackage{algorithmic}
\usepackage{graphicx}
\usepackage{subfig}
\usepackage{textcomp}
\usepackage{physics}
\usepackage{bm}

\newtheorem{theoremx}{Theorem}
\newtheorem{lemmax}{Lemma}
\newtheorem{corollaryx}{Corollary}

\newtheorem{propositionx}{Proposition}

\DeclareDocumentCommand\diag{}{\opbraces{\operatorname{diag}}}
\DeclareDocumentCommand\bdiag{}{\opbraces{\operatorname{bdiag}}}

\title{\LARGE \bf On the Stability of Consensus Control under Rotational Ambiguities
}

\author{Zhonggang Li, Changheng Li, Raj Thilak Rajan
\thanks{This work is funded by the Sensor AI Lab, under the AI Labs program of Delft University of Technology.The authors are with Signal Processing Systems, Faculty of EEMCS, Delft University of Technology, 2628 CD Delft, The Netherlands (email: \{z.li-22, l.li-7, r.t.rajan\}@tudelft.nl)}%
}

\begin{document}

\maketitle
\thispagestyle{empty}
\pagestyle{empty}

\begin{abstract}
Consensus control of multiagent systems arises in various robotic applications such as rendezvous and formation control. For example, to compute the control inputs of individual agents, the difference in the positions in aligned coordinate frames i.e., the pairwise displacements are typically measured. However, the local coordinate frames might be subject to rotational ambiguities, such as a rotation or a reflection, particularly if the positions of the agent are not directly observed but reconstructed from e.g. pairwise Euclidean distances. This rotational ambiguity causes stability issues in practice, as agents have rotated perceptions of the environment. In this work, we conduct a thorough analysis of the stability in the presence of rotational ambiguities in several scenarios including e.g., proper and improper rotation, and the homogeneity of rotations. We give stability criteria and stability margin on the rotations, which are numerically verified with two traditional examples of consensus control.
\end{abstract}

\section{Introduction}\label{sec:introduction}


Consensus algorithms are essential in modern distributed systems across various fields, e.g., in wireless sensor networks for decentralized data aggregation and distributed estimation \cite{shi2019fast, bertrand2011consensus}, and distributed power systems e.g., microgrids \cite{fan2019consensus, li2017consensus}. These algorithms are also widely applied to the distributed control of mobile multiagent systems in swarm robotic applications e.g., flocking of free robots \cite{bhowmick2016flocking}, rendezvous control of dispersed robots into a common location \cite{parasuraman2018multipoint}, or formation control to achieve a desired geometrical pattern \cite{olfati2004consensus, lin2015necessary}. Consensus algorithms typically involve interagent interactions that usually require exchanging the local states over a communication network. In multiagent control applications, the local states are usually the kinematics of agents, e.g., positions, velocities, etc. The local control input is then computed based on the difference in the agent's and its neighbors' kinematics. The most common practice is to measure the absolute kinematics and share the information through a communication network \cite{ren2008distributed}, which implies the need for global navigation satellite systems (GNSS) or other global positioning systems. However, global positioning is tough to acquire in many environments such as indoor applications \cite{garcia2021cooperative}, subterranean operations \cite{otsu2020supervised}, and extraterrestrial explorations \cite{thangavelautham2020autonomous}, which makes relative positioning and navigation increasingly appealing. 

Recently, there has been some attention on consensus control by reconstructing positions using a set of pairwise distance measurements \cite{han2018integrated,fang2023integrated}, which is sometimes referred to as relative localization. The benefit is not only that distance measurements can be acquired fairly easily with low-cost sensors, but also the solution is distributable across the agents \cite{liu2020distributed}. However, some global information is lost in the scalar-valued distance encodings, and hence ambiguities emerge in the reconstruction process. These ambiguities are typically rigid transformations such as translations, rotations, etc, as the distances are invariant to these transformations \cite{dokmanic2015euclidean}, which also exist in higher-order kinematics \cite{rajan2019relative}. Note that rotational ambiguities can be considered unaligned local frames compared to global reference frames. An industry-standard procedure is to deploy anchors with global information to correct these ambiguities \cite{dokmanic2015euclidean}, which is assumed by default in many algorithms \cite{han2018integrated}. However, in anchorless networks, the influence of ambiguities has to be addressed but is rarely studied in the literature. General concerns about the ambiguities for consensus control include the stability or the convergence of the consensus system under ambiguities, the invariance of the steady state or the equilibrium of the system, and the convergence speed.

In this paper, we focus on consensus control systems that rely only on positions for control inputs. We acknowledge that the steady state of consensus control is invariant to ambiguities. Then we focus on the stability under rotational ambiguities modeled by rotation matrices by analyzing the eigenvalue distributions of such a linear dynamical system. In Section \ref{sec: fundamentals}, the preliminaries of consensus control are introduced and a general modeling of the ambiguous system is given. In Section \ref{sec: general conditions}, conditions of stability under the general ambiguous model are given as well as several specific scenarios including planar and spatial cases, homogeneity of local ambiguities, and proper and improper rotations. The established theories are then verified through two traditional consensus control cases, namely rendezvous control and affine formation control in Section \ref{sec: sim} before conclusions and insights of future research are given in Section \ref{sec: conclusions}.   

\textit{Notations.} 
Vectors and matrices are represented by lowercase and uppercase boldface letters respectively such as $\bm{a}$ and $\bm{A}$. Sets and graphs are represented using calligraphic letters e.g., $\mathcal{A}$. Vectors of length $N$ of all ones and zeros are denoted by $\bm{1}_N$ and $\bm{0}_N$  respectively. An identity matrix of size $N$ is denoted by $\bm{I}_N$. The Kronecker product is denoted by $\otimes$. A set of $D$-dimensional real-valued vectors and $D\times D$ real-valued matrices are denoted by $\mathbb{R}^D$ and $\mathbb{R}^{D\times D}$, respectively. The operator $\sigma(\cdot)$ gives the set of eigenvalues of a diagonalizable matrix. The $\diag(\cdot)$ operator creates a diagonal matrix from a vector and $\bdiag(\bm{A}_1,...,\bm{A}_N)$ creates a block diagonal matrix from matrices $\bm{A}_i$ for $i \in \mathcal{V}$. We define the operator $\gamma\qty(\bm{A}) = \frac{1}{2}(\bm{A}+\bm{A}^T)$ for a square matrix $\bm{A}$. All eigenvalues of $\bm{A}$ has strictly negative real parts if $\gamma\qty(\bm{A})$ is negative-definite, i.e., $\gamma\qty(\bm{A}) \prec 0$ \cite{weissteinPdm}.
\section{Fundamentals and Problem Formulation}\label{sec: fundamentals}

We consider a multiagent system where $N$ agents are interconnected through an undirected connected graph $\mathcal{G}=\qty(\mathcal{V},\mathcal{E})$ where the set of vertices is $\mathcal{V} = \qty{1,2,..,N}$ and the set of edges is $\mathcal{E}=\mathcal{V}\times\mathcal{V}$. The set of neighbors of agent $i$ is denoted as $\mathcal{N}_i=\qty{j\in\mathcal{V}: \qty(i,j)\in\mathcal{E}}$. A generalized Laplacian matrix $\bm{L}\in\mathbb{R}^{N\times N}$ of a graph is defined as 
\begin{equation}\label{equ: def Lap}
    [\bm{L}]_{ij}=\begin{cases}
  \sum_{j\in\mathcal{N}_i}l_{ij}& \text{if } i=j \\
  -l_{ij}& \text{if }i\neq j \text{ and } j\in\mathcal{N}_i\\
  0&  \text{otherwise}
\end{cases},
\end{equation}
where $l_{ij}\in\bm{R}, \forall \qty(i,j)\in\mathcal{E}$ are the weights associated with the edges. If $l_{ij}=1$, $\bm{L}$ reduces to the most common form seen in the literature.  $\bm{L}$ is a symmetric positive-semidefinite matrix with $\bm{L}\bm{1}=\bm{0}$. A Laplacian typically has rank $N-1$ but can also have a lower rank $P<N$ if defined as (\ref{equ: def Lap}) for a connected graph. An example is a stress matrix \cite{alfakih2011bar}, which is used in some formation control problems \cite{lin2015necessary}. The eigenvalue decomposition of $\bm{L}$ is
\begin{equation} \label{equ: evd Lap}
        \bm{L}= \begin{bmatrix}
         \bm{U}_1 &\bm{U}_2
        \end{bmatrix}\begin{bmatrix}
          \bm{\Lambda}& \\
          &\bm{0}
        \end{bmatrix}\begin{bmatrix}
         \bm{U}_1^T\\
        \bm{U}_2^T
        \end{bmatrix},
    \end{equation}
where $\bm{\Lambda} = \diag([\lambda_1,...,\lambda_{P}])$ is positive-definite, and $\bm{U}_1\in\mathbb{R}^{N\times P}$ and $\bm{U}_2\in\mathbb{R}^{N\times (N-P)}$ span the range and the nullspace of $\bm{L}$, respectively.

\subsection{Consensus Systems}
A typical consensus system is a dynamic system
\begin{equation}\label{equ: consensus system}
        \dot{\bm{z}} = -(\bm{L}\otimes\bm{I}_D)\bm{z}= -\bm{\tilde{L}}\bm{z},
\end{equation}
where $\bm{z} = \qty[\bm{z}_1^T,...,\bm{z}_N^T]^T\in\mathbb{R}^D$ is the global state, in which $\bm{z}_i\in\mathbb{R}^D$ for $i \in \mathcal{V}$ are the states of each agents, and $\dot{\bm{z}}$ is the first order derivative of $\bm{z}$. We consider $D=2,3$ for practical applications. Note that an equilibrium point of (\ref{equ: consensus system}) resides in the nullspace of $\bm{\tilde{L}}$ which admits a decomposition
\begin{equation}\label{equ: evd Lap kron}
        \tilde{\bm{L}}=(\bm{L}\otimes\bm{I}_D) = \begin{bmatrix}
         \tilde{\bm{U}}_1 &\tilde{\bm{U}}_2
        \end{bmatrix}\begin{bmatrix}
          \tilde{\bm{\Lambda}}& \\
          &\bm{0}
        \end{bmatrix}\begin{bmatrix}
         \tilde{\bm{U}}_1^T\\
        \tilde{\bm{U}}_2^T
        \end{bmatrix},
    \end{equation}
where $\tilde{\bm{\Lambda}} = \bm{\Lambda}\otimes\bm{I}_D$ are the non-zero eigenvalues and $\tilde{\bm{U}}_1 = \bm{U}_1\otimes\bm{I}_D$ are valid as eigenvectors according to Lemma \ref{lmm: kron eigs} in Appendix \ref{apdx: alt proofs}. 



\subsection{Problem Formulation}
We now consider (\ref{equ: consensus system}) under ambiguity i.e.,
\begin{equation}\label{equ: general ambiguous}
    \dot{\bm{z}} = -\tilde{\bm{H}}(\bm{L}\otimes\bm{I}_D)\bm{z} 
    = -\tilde{\bm{H}}\tilde{\bm{L}}\bm{z},
\end{equation} where $\tilde{\bm{H}} = \bdiag\qty(\bm{H}_1,...,\bm{H}_N)\in\mathbb{R}^{ND\times ND}$ is the global ambiguity matrix and $\bm{H}_i\in\mathbb{R}^{D\times D} \ \forall i\in \mathcal{V}$ represents the local ambiguity. In general, $\bm{H}_i$ can model a wide range of ambiguities, but we focus on rotational ambiguities in this work. Observe that equilibrium points of (\ref{equ: consensus system}) and (\ref{equ: general ambiguous}) are the same, which means that the steady-state solutions are unchanged if the system is stable. In this work, we aim to investigate the stability of the ambiguous system (\ref{equ: general ambiguous}) and give conditions on $\tilde{\bm{H}}$ to yield a stable system in various practical scenarios.

The ambiguities considered in this work include proper and improper rotations that are orthogonal matrices with determinants $1$ and $-1$, respectively. Note that the common reflection matrices are special cases of improper rotations. We define $\bm{H}_i = \mathcal{R}(\theta_i)\bm{T}_i\in\mathbb{R}^{D\times D}$ for $i \in \mathcal{V}$ where $\mathcal{R}(\theta_i)$ gives a rotation matrix using an angle $\theta_i\in(-\pi,\pi]$ and $\bm{T}_i$ is a diagonal matrix of $\pm1$s and $\bm{H}_i$ is proper if $\bm{T}_i = \bm{I}_D$. The structures of proper and improper rotations for $D=2$ and $D=3$ are shown in Table \ref{tab: rot mats}, where without the loss of generality we assume the rotation angle $\theta_i$ for $D=3$ is around the ``$z$-axis" \cite{horn2012matrix}. 


\begin{table}[t]
\scriptsize
	\caption{Rotational Ambiguities $\bm{H}_i$}
	\label{tab: rot mats}
	\centering
	\begin{tabularx}{\linewidth}{l l l}
		\toprule
		$D$ & Proper rotation $\mathcal{R}(\theta_i)$ & Improper rotation $\mathcal{R}(\theta_i)\bm{T}_i$\\
		\midrule
		$2$ & $\begin{bmatrix}
  \cos \theta_i& -\sin \theta_i\\
 \sin \theta_i &\cos \theta_i
\end{bmatrix}$ &$\begin{bmatrix}
  \cos \theta_i& -\sin \theta_i\\
 \sin \theta_i &\cos \theta_i
\end{bmatrix}\begin{bmatrix}
  1 & \\
  &-1
\end{bmatrix}$ \vspace{5pt} \\ 
		
		$3$ &
  $\begin{bmatrix}
  \cos \theta_i& -\sin \theta_i & 0\\
 \sin \theta_i &\cos \theta_i & 0\\
 0 & 0 & 1
\end{bmatrix}$ & $\begin{bmatrix}
  \cos \theta_i& -\sin \theta_i & 0\\
 \sin \theta_i &\cos \theta_i & 0\\
 0 & 0 & 1
\end{bmatrix}\begin{bmatrix}
  1 & & \\
 &1 & \\
 &  & -1
\end{bmatrix}$\\
		\bottomrule
	\end{tabularx}
\end{table}

\section{Stability under ambiguities }\label{sec: general conditions}
We now give general conditions for the ambiguous system (\ref{equ: general ambiguous}) to be stable, followed by an investigation on stability under specific types of ambiguities in different scenarios based on these general conditions. We introduce the following lemma used in the proof of general stability criteria.
\begin{lemmax}\label{lmm: adapted Hurwitz}
(\textit{Eigenvalue distributions of a product of matrices})
    Given a product of two square matrices $\bm{A} = \bm{G}\bm{Q}$ with $\bm{Q}$ being symmetric positive-definite, $\gamma(\bm{A})\prec 0$ if and only if $\gamma(\bm{G})\prec 0$.
\end{lemmax}
\begin{proof}
    Since $\bm{Q}$ is symmetric postiive-definite, $\bm{A} = \bm{G}\bm{Q}$ and $\bm{G} = \bm{A}\bm{Q}^{-1}$ both satisfied the defined decomposition. Then $\gamma(\bm{A})\prec 0$ and $\gamma(\bm{G})\prec 0$ are equivalent using the direct conclusion from \cite[Theorem 3.1]{DUAN1998Hurwitz}.
\end{proof}

\subsection{General conditions for Stability} We first acknowledge that the unambiguous system (\ref{equ: consensus system}) is globally and exponentially stable, which is an extension of the conclusion in \cite{saber2003consensus}. Observe that the $PD$ nonzero eigenvalues of $-\tilde{\bm{L}}$ are strictly negative, which means that given any arbitrary initialization $\bm{z}(t)$ at $t=0$, $\bm{z}(t)$ will exponentially converge to an equilibrium point in the nullspace of $\tilde{\bm{L}}$, i.e., $\lim_{t\rightarrow\infty}\norm{\bm{z}(t)-\bm{z}_e}=0$ where $\bm{z}_e$ lives in the nullspace of $(\tilde{\bm{L}})$. Therefore, to ensure the stability of the ambiguous system (\ref{equ: general ambiguous}), the nonzero eigenvalues of the product $-\tilde{\bm{H}}\tilde{\bm{L}}$ should have strictly negative real parts, which is described in the next theorem.

\begin{theoremx}\label{thm: stability ambi}
(\textit{General stability criteria for systems with rotational ambiguity}) Let $\tilde{\bm{U}}_1$ span the range of $\tilde{\bm{L}}$, then the ambiguous system (\ref{equ: general ambiguous}) is globally and exponentially stable if and only if $\gamma(\tilde{\bm{U}}_1^T\tilde{\bm{H}}\tilde{\bm{U}}_1) \succ 0$. 
\end{theoremx}
\begin{proof}
Using properties of eigenvalues of products of matrices and decomposition (\ref{equ: evd Lap kron}) it holds that
    \begin{align}
        &\sigma\qty(-\tilde{\bm{H}}\tilde{\bm{L}})= \sigma\bigg(-\tilde{\bm{H}}\begin{bmatrix}
         \tilde{\bm{U}}_1 &\tilde{\bm{U}}_2
        \end{bmatrix}\begin{bmatrix}
          \tilde{\bm{\Lambda}}& \\
          &\bm{0}
        \end{bmatrix}\begin{bmatrix}
         \tilde{\bm{U}}_1^T\\
        \tilde{\bm{U}}_2^T
        \end{bmatrix}\bigg) \notag \\
        &=\sigma\bigg(-\begin{bmatrix}
         \tilde{\bm{U}}_1^T\\
        \tilde{\bm{U}}_2^T
        \end{bmatrix}\tilde{\bm{H}}\begin{bmatrix}
         \tilde{\bm{U}}_1 &\tilde{\bm{U}}_2
        \end{bmatrix}\begin{bmatrix}
          \tilde{\bm{\Lambda}}& \\
          &\bm{0}
        \end{bmatrix}\bigg)\label{equ: circ eigs}.
    \end{align}
    Observe that $-\tilde{\bm{U}}_1^T\tilde{\bm{H}}\tilde{\bm{U}}_1\tilde{\bm{\Lambda}}$ corresponds to the nonzero eigenvalues. As such, given $\tilde{\bm{\Lambda}}$ is positive definite, the system is stable if and only if  $\gamma(\tilde{\bm{U}}_1^T\tilde{\bm{H}}\tilde{\bm{U}}_1) \succ 0$, according to Lemma \ref{lmm: adapted Hurwitz}. Moreover, note that $\tilde{\bm{U}}_1$ might be up to an orthogonal transformation $\tilde{\bm{P}}$ than assumed for (\ref{equ: evd Lap kron}), which yields $\gamma(\tilde{\bm{P}}^T\tilde{\bm{U}}_1^T\tilde{\bm{H}}\tilde{\bm{U}}_1\tilde{\bm{P}})\succ 0$, which is the same as $\gamma(\tilde{\bm{U}}_1^T\tilde{\bm{H}}\tilde{\bm{U}}_1) \succ 0$ due to congruence.
\end{proof}

\begin{corollaryx}\label{crly: pd sandwich}
(\textit{Sufficient condition for stability})
The ambiguous system (\ref{equ: general ambiguous}) is globally and exponentially stable if $\gamma(\tilde{\bm{H}}) \succ 0$.
\end{corollaryx}
\begin{proof}
    It can be verified that $\tilde{\bm{U}}_1^T\gamma(\tilde{\bm{H}})\tilde{\bm{U}}_1 = \gamma(\tilde{\bm{U}}_1^T\tilde{\bm{H}}\tilde{\bm{U}}_1)$. As such, it holds that $\gamma(\tilde{\bm{U}}_1^T\tilde{\bm{H}}\tilde{\bm{U}}_1)\succ 0$ given that $\gamma(\tilde{\bm{H}}) \succ 0$ and  a full-rank $\tilde{\bm{U}}_1$ \cite{horn2012matrix}, which leads to a stable system according to Theorem \ref{thm: stability ambi}.
\end{proof}

\noindent Corollary \ref{crly: pd sandwich} further narrows down the conditions for stability directly regarding the ambiguity matrix $\tilde{\bm{H}}$. Note that this condition is sufficient but not necessary as $\tilde{\bm{U}}_1$ is a tall matrix.


\subsection{Stability under Homogeneous Ambiguities}\label{sec: homo ambiguities}


In some scenarios, all agents can share homogeneous local ambiguities, i.e., $\bar{\bm{z}}_i = \bm{H}\bm{z}_i, \forall \bm{H}_i=\bm{H}, i\in\mathcal{V}$, and subsequently (\ref{equ: general ambiguous}) is  
\begin{equation}\label{equ: homo rot}
    \dot{\bm{z}} 
    = -(\bm{I}_N\otimes \bm{H})(\bm{L}\otimes\bm{I}_D)\bm{z}
    = -\tilde{\bm{H}}\tilde{\bm{L}}\bm{z}.
\end{equation}

We now discuss the stability conditions for both proper and improper rotations $\bm{H}$, in Theorems \ref{thm: homo proper rot}  and \ref{thm: homo improp rot}, respectively. Alternative proofs are presented in Appendix \ref{apdx: alt proofs}.
\begin{theoremx}\label{thm: homo proper rot}
(\textit{Stability under homogeneous and proper rotations})
The ambiguous system (\ref{equ: homo rot}) is globally and exponentially stable if and only if the rotation angle $\theta$ of a proper rotation $\bm{H} = \mathcal{R}\qty(\theta)$ lies within the range $\theta\in(-\frac{\pi}{2},\frac{\pi}{2})$.
\end{theoremx}
\begin{proof}
    Recollect from Theorem \ref{thm: stability ambi} that $\gamma(\tilde{\bm{U}}_1^T\tilde{\bm{H}}\tilde{\bm{U}}_1) \succ 0$ is required for stability.  In case of $D=2$, using the rotation matrix from Table \ref{tab: rot mats}, we can verify that
    \begin{align}\gamma(\tilde{\bm{U}}_1^T\tilde{\bm{H}}\tilde{\bm{U}}_1)
        =\tilde{\bm{U}}_1^T\cos\theta\bm{I}_{ND}\tilde{\bm{U}}_1 =\cos\theta \bm{I}_{PD}, \label{equ: homo connonical}
    \end{align} and for $D=3$, substituting for $\tilde{\bm{U}}_1$ and $\tilde{\bm{H}}$,  we have 
    \begin{align}\label{equ: reduced GED}
    \gamma(\tilde{\bm{U}}_1^T\tilde{\bm{H}}\tilde{\bm{U}}_1) 
    = \gamma(\bm{I}_{P}\otimes\bm{H}) 
    = \bm{I}_P\otimes\diag(\begin{bmatrix}
  \cos\theta& \cos\theta &1
\end{bmatrix}).
    \end{align} Both (\ref{equ: homo connonical}) and (\ref{equ: reduced GED}) are block diagonal matrices, which are positive-definite if and only if $\theta\in(-\frac{\pi}{2},\frac{\pi}{2})$. Hence proved.
    \hfill
\end{proof}
A numerical example is shown in Fig. \ref{fig: eig homo prop}, where we present the eigenvalue distribution using the graph in Fig. \ref{fig: used graphs} (a) and its standard Laplacian as defined in (\ref{equ: def Lap}). Fig. \ref{fig: eig homo prop} (a) shows the rotations of eigenvalues from lemma \ref{lmm: alt proper rot}, where the eigenvalues of $-\bm{L}$ originally lie on the imaginary axis but then are rotated by $\bm{H}$. Then $\theta\in(-\frac{\pi}{2},\frac{\pi}{2})$ ensures that they stay on the left complex plane. Fig. \ref{fig: eig homo prop} (b) shows all the eigenvalues of $\gamma\qty(\tilde{\bm{U}}_1^T\tilde{\bm{H}}\tilde{\bm{U}}_1)$ across a spectrum of $\theta$. The region $\theta\in(-\frac{\pi}{2},\frac{\pi}{2})$ (shaded in purple) guarantees positive eigenvalues meaning $\gamma\qty(\tilde{\bm{U}}_1^T\tilde{\bm{H}}\tilde{\bm{U}}_1)\succ 0$, which verifies Theorem \ref{thm: homo proper rot}.

\begin{figure}[t]
	\centering	

        \subfloat[\scriptsize  Eigenvalues of $-(\bm{I}_N\otimes\bm{H})(\bm{L}\otimes\bm{I}_D)=-(\bm{L}\otimes\bm{H})$]{\raisebox{0ex}
		{\includegraphics[width=0.35\textwidth]{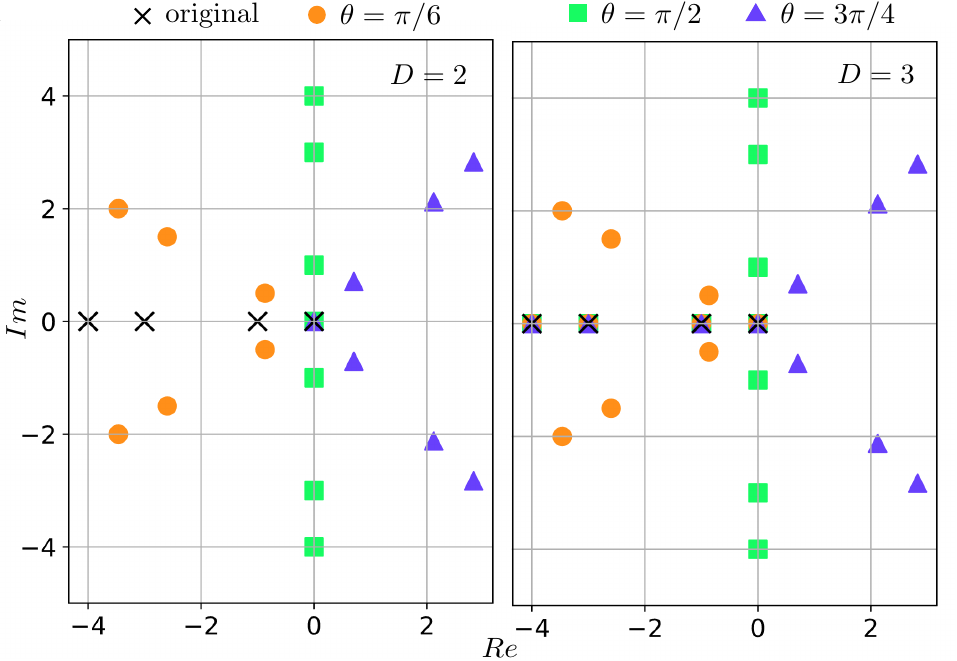}}%
	}
	\hspace{1ex}
	\subfloat[\scriptsize  Eigenvalues of $\gamma\qty(\tilde{\bm{U}}_1^T\tilde{\bm{H}}\tilde{\bm{U}}_1)$]{\raisebox{0ex}
		{\includegraphics[width=0.35\textwidth]{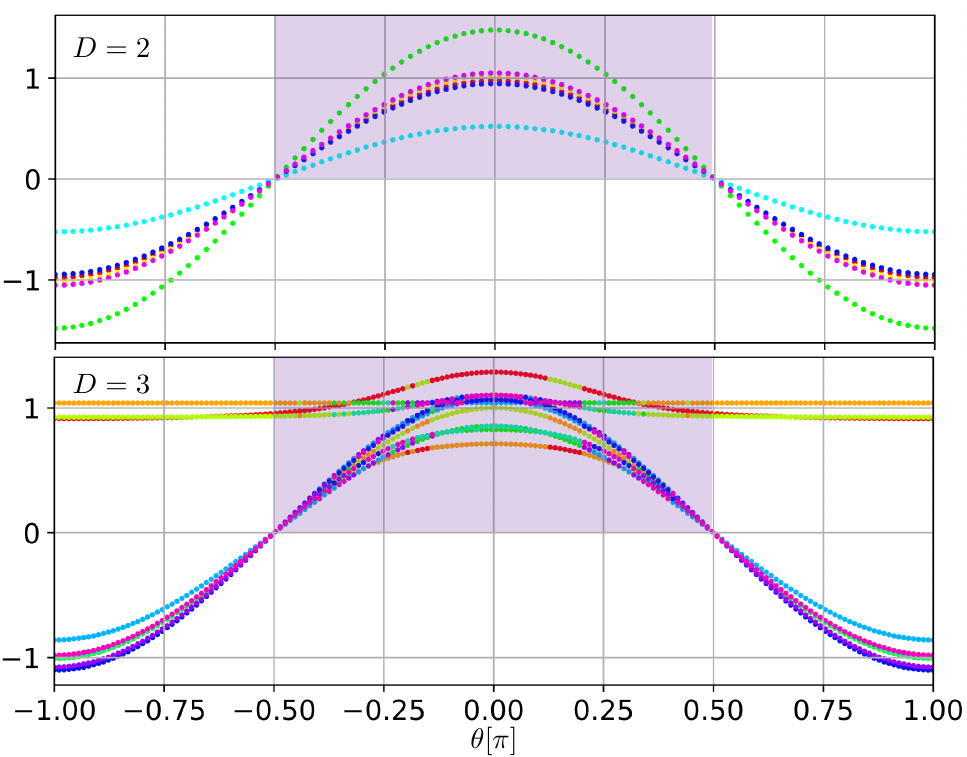}}%
	}

	\caption{A numerical example of the eigenvalues under homogeneous and proper rotations.}
	\label{fig: eig homo prop}
		
\end{figure}

\begin{theoremx}\label{thm: homo improp rot}
(\textit{Stability under homogeneous and improper rotations})
The ambiguous system (\ref{equ: homo rot}) is unstable under improper rotations $\bm{H}$.
\end{theoremx}

\begin{proof}
    We show that $\gamma(\tilde{\bm{U}}_1^T\tilde{\bm{H}}\tilde{\bm{U}}_1)\succ 0$ is impossible. Recall from (\ref{equ: reduced GED}) that $\tilde{\bm{U}}_1^T\tilde{\bm{H}}\tilde{\bm{U}}_1$ can be simplified to $\bm{I}_P\otimes\bm{H}$. 
    Due to the negative determinant of an improper rotation, $\gamma(\bm{H})$ is never positive-definite. As such, $\gamma(\bm{I}_P\otimes\bm{H})$ is not positive-definite either and hence the system is unstable under improper rotations independent of dimension $D$.  
\end{proof} 
Fig. \ref{fig: eig homo improp} shows the same example as Fig. \ref{fig: eig homo prop} but with improper rotations, where there are always positive eigenvalues present in (a) and there is no region in (b) where $\gamma\qty(\tilde{\bm{U}}_1^T\tilde{\bm{H}}\tilde{\bm{U}}_1)\succ 0$. Hence, system (\ref{equ: homo rot}) is not stable under improper rotations.

\begin{figure}[t]
	\centering	

        \subfloat[\scriptsize Eigenvalues of $-(\bm{I}_N\otimes\bm{H})(\bm{L}\otimes\bm{I}_D) = -(\bm{L}\otimes\bm{H})$]{\raisebox{0ex}
		{\includegraphics[width=0.35\textwidth]{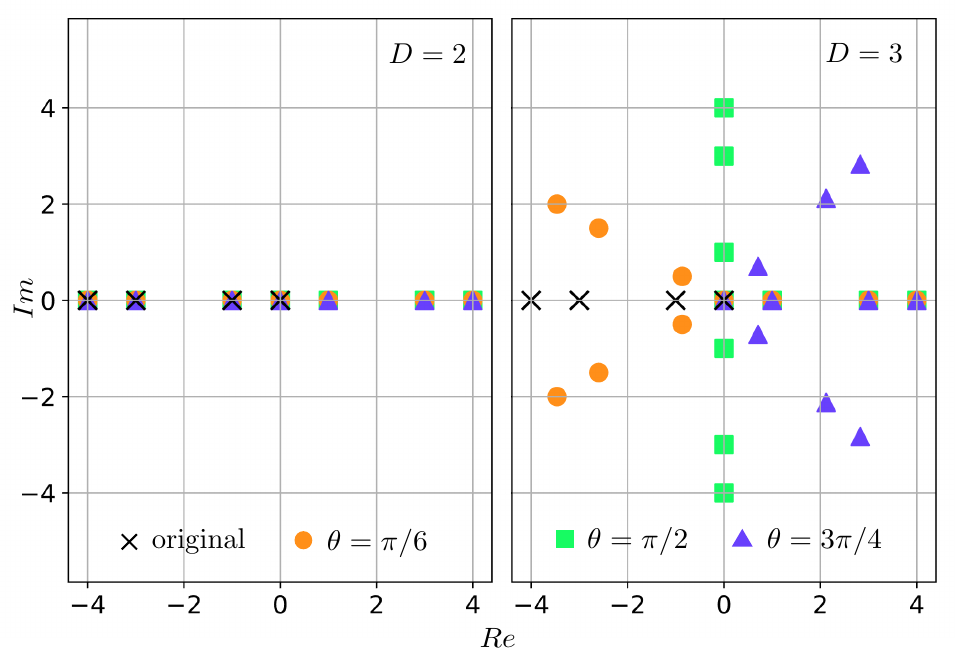}}%
	}
	\vspace{1mm}
	\subfloat[\scriptsize Eigenvalues of $\gamma\qty(\tilde{\bm{U}}_1^T\tilde{\bm{H}}\tilde{\bm{U}}_1)$]{\raisebox{0ex}
		{\includegraphics[width=0.35\textwidth]{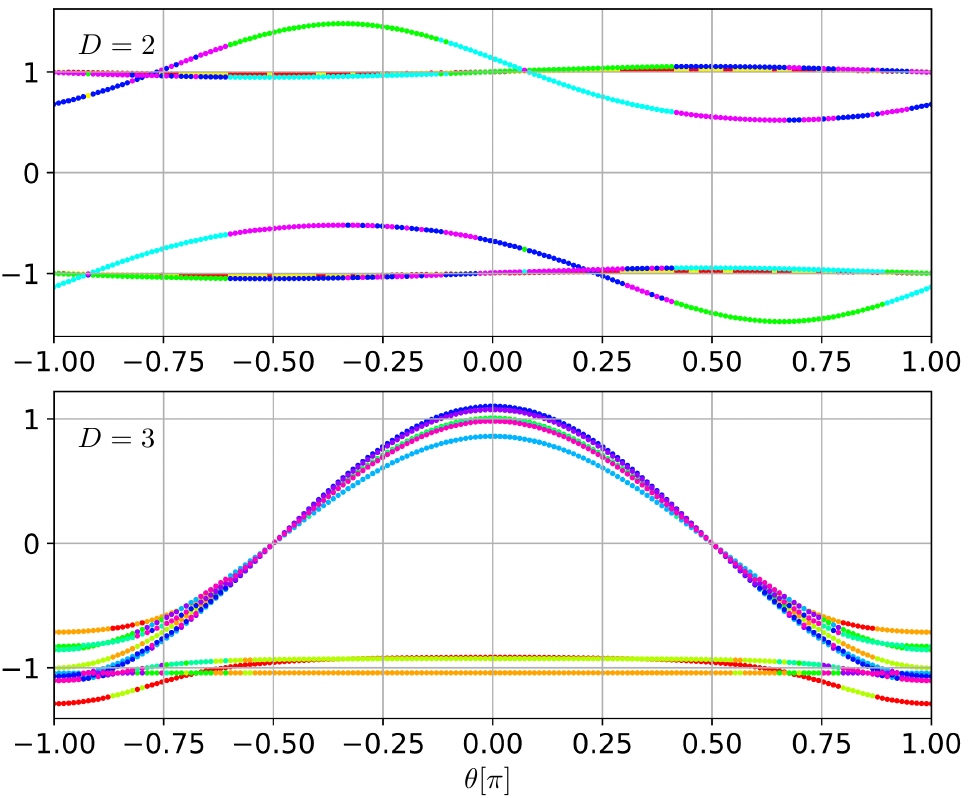}}%
	}

	\caption{A numerical example of the eigenvalues under homogeneous and improper rotations.}
	\label{fig: eig homo improp}
		
\end{figure}

\subsection{Stability under Heterogeneous Ambiguities}\label{sec: hetero stability}
In some distributed cases, each agent might have a different perception of the other agents' positions which we model as heterogeneous ambiguities, i.e., $\bar{\bm{z}}_j = \bm{H}_i\bm{z}_j, \forall j\in\{i\}\cup\mathcal{N}_i$ for agent $i$. The general ambiguous system (\ref{equ: general ambiguous}) is now \begin{align}\label{equ: hetero rot}
    \dot{\bm{z}} 
    = -\bdiag\qty(\bm{H}_1,...,\bm{H}_N)(\bm{L}\otimes\bm{I}_D)\bm{z},
    = -\tilde{\bm{H}}\tilde{\bm{L}}\bm{z}.
\end{align} There are three potential cases under this model: (a) all proper rotations, (b) all improper rotations, or (c) a mixture of proper and improper rotations across the agents. 

We show that, unlike the homogeneous scenario, in the case of heterogeneous proper rotations, $\theta_i\in(-\frac{\pi}{2},\frac{\pi}{2}) \ \forall i\in\mathcal{V}$ is a sufficient but not necessary condition for a stable system under proper rotations $\bm{H}_i \ \forall i\in\mathcal{V}$.
\begin{theoremx}\label{thm: hetero proper rot}
(\textit{Stability under heterogeneous and proper rotations})
The ambiguous system (\ref{equ: hetero rot}) is globally and exponentially stable if the corresponding rotation angles $\theta_i$ of a local proper rotation $\bm{H}_i = \mathcal{R}(\theta_i)$ lie within range $\theta_i\in(-\frac{\pi}{2},\frac{\pi}{2}) \ \forall i\in\mathcal{V}$.
\end{theoremx}
\begin{proof} Recollect from Lemma \ref{crly: pd sandwich} that the sufficient condition for a stable system is to have a generalized positive-definite $\tilde{\bm{H}}$. For $D=2$ and $D=3$, we observe that
\begin{align}
    \gamma(\tilde{\bm{H}}) &= \bdiag(\cos\theta_1\bm{I}_2,...,\cos\theta_N\bm{I}_2),
    \quad (D=2)
    \label{equ: blk cos pd} \\
    \gamma(\tilde{\bm{H}}) &=
    \bdiag\qty(\gamma\qty(\mathcal{R}(\theta_1)),...,\gamma\qty(\mathcal{R}
    (\theta_N))),
    \quad (D=3)
\end{align} which are both diagonal matrices, with the property $\gamma(\tilde{\bm{H}})\succ 0$ if $\theta_i\in(-\frac{\pi}{2},\frac{\pi}{2}) \ \forall i\in\mathcal{V}$, for a stable system (\ref{equ: hetero rot}). \end{proof}

We use the same settings as in Fig. \ref{fig: eig homo prop} but with heterogeneous rotations to show another numerical example in Fig. \ref{fig: eig hetero} (a) where the smallest eigenvalue is shown in a heatmap across a spectrum of $\theta_3$ and $\theta_4$. Observe that the white bounding box, inside of which are eigenvalues greater than zero, is bigger than the area boxed by $(-\frac{\pi}{2},\frac{\pi}{2})$ in yellow. This shows that there exists $\theta_i\notin(-\frac{\pi}{2},\frac{\pi}{2})$ that still entails $\gamma(\tilde{\bm{U}}_1^T\tilde{\bm{H}}\tilde{\bm{U}}_1)\succ 0$ i.e.,  a stable system, which verifies the sufficiency but not necessity of Theorem \ref{thm: hetero proper rot}.



\begin{figure}[t]
    \centering	

    \captionsetup[subfloat]{width=\linewidth}  

    \subfloat[The smallest eigenvalue of $\gamma(\tilde{\bm{U}}_1^T\tilde{\bm{H}}\tilde{\bm{U}}_1)$ across a spectrum of $\theta_3$ and $\theta_4$. In this case, we set $\theta_1=\theta_2=0$. Zero values are highlighted in white, and the region between $-\frac{\pi}{2}$ and $\frac{\pi}{2}$ is highlighted by a yellow dashed box.]
    {\includegraphics[width=0.3\textwidth]{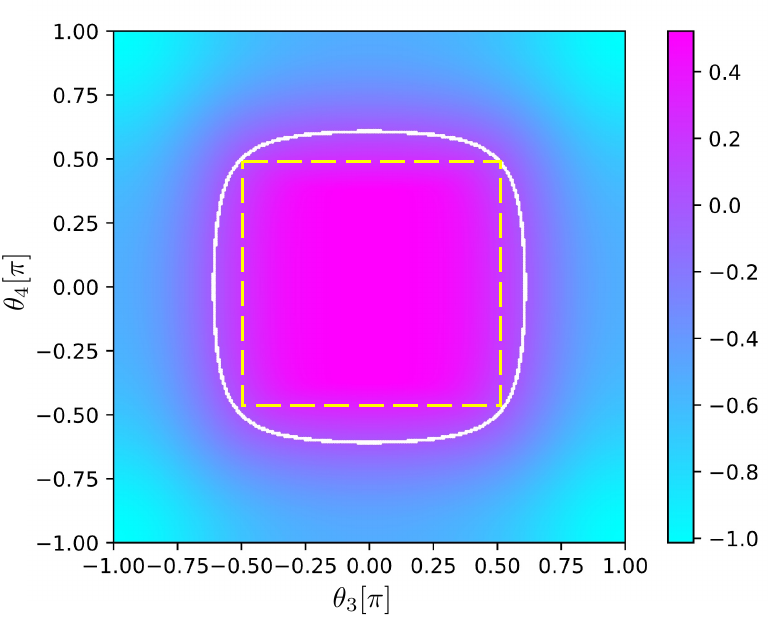}}

    \vspace{1mm}

    \subfloat[Eigenvalues of$\gamma(\tilde{\bm{U}}_1^T\tilde{\bm{H}}\tilde{\bm{U}}_1)$. In this case, $\theta_1=\theta_2=\theta_3=0$ and $\bm{H}_i$ for $i=1,2,3$ are proper, while only $\bm{H}_4$ is improper.]
    {\includegraphics[width=0.4\textwidth]{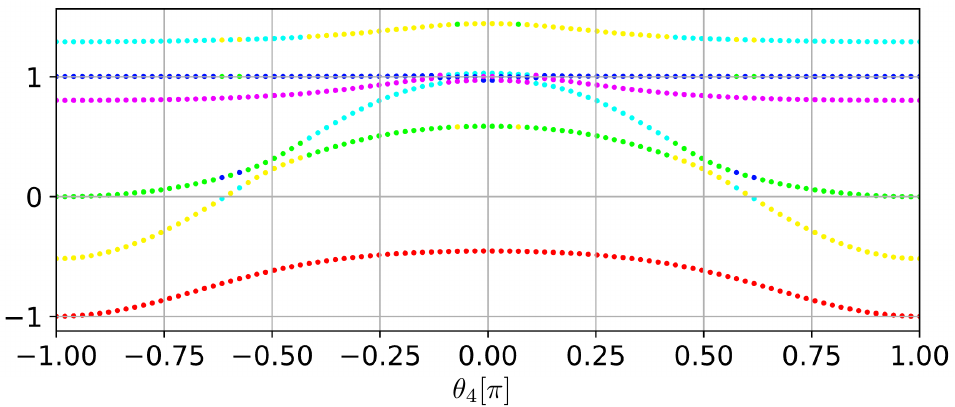}}

    \caption{A numerical example of the eigenvalues under heterogeneous rotations.}
    \label{fig: eig hetero}
\end{figure}

We now discuss the cases where one or more improper rotations appear among all agents. We make a proposition and give intuitive reasoning, which is verified with numerical examples and simulations in later sections.
\begin{propositionx}\label{prop: hetero mix rot}
(\textit{Instability under mixture of rotations})
The ambiguous system (\ref{equ: hetero rot}) is unstable if there exists $i$ such that $\bm{H}_i$ is an improper rotation.
\end{propositionx} A special case of this scenario is that $\bm{H}_i\forall i \in \mathcal{V}$ are homogeneous improper rotations, which is proven to be unstable in Theorem \ref{thm: homo improp rot}. The more general case from (\ref{equ: general ambiguous}) is
\begin{align}
    &\tilde{\bm{H}}(\bm{L}\otimes\bm{I}_D) = \notag\\
      &\begin{bmatrix}
 \mathcal{R}(\theta_1) &  & \\
  & \ddots & \\
  &  & \mathcal{R}(\theta_N)
\end{bmatrix}
\begin{bmatrix}
 \bm{T}_1 &  & \\
  & \ddots & \\
  &  & \bm{T}_N
\end{bmatrix}(\bm{L}\otimes\bm{I}_D),
\end{align}
where $\bm{T}_i=\bm{I}$ if $\bm{H}_i$ for any $i \in \mathcal{V}$ is a proper rotation. We can consider $\bdiag\qty(\bm{T}_1,...,\bm{T}_N)(\bm{L}\otimes\bm{I}_D)$ a new Laplacian matrix where some rows are negated if certain $\bm{H}_i$ is not proper. This new Laplacian matrix is no longer symmetric positive semi-definite in general and yields an unstable system regardless of what the proper rotation part is. The example in Fig. \ref{fig: eig hetero} (b) shows that as long as one $\bm{H}_i$ for any $i \in \mathcal{V}$ is improper, $\gamma(\tilde{\bm{U}}_1^T\tilde{\bm{H}}\tilde{\bm{U}}_1)$ is not positive-definite even if all the other agents are unambiguous. 


\section{Examples and Simulations}\label{sec: sim}
In this section, we verify our theorems and proposition with two algorithms under consensus frameworks, namely, the rendezvous control \cite{amirkhani2022consensus} and distributed formation control \cite{zhao2018affine}. The graphs used for each case are shown in Fig. \ref{fig: used graphs}, where there are $4$ and $7$ nodes, respectively. Agents in these algorithms are assumed to adopt single integrator dynamics $\dot{\bm{z}}_i=\bm{u}_i\ \forall i\in\mathcal{V}$. In both scenarios, the global dynamical model (\ref{equ: consensus system}) translates to local control input as follows,
\begin{equation}
    \bm{u}_i = -\sum_{j\in\mathcal{N}_i}l_{ij}(\bm{z}_i-\bm{z}_j),
\end{equation}
where $l_{ij}$ are the Laplacian weights in (\ref{equ: def Lap}). For a consensus system (\ref{equ: consensus system}), the equilibrium points are not unique i.e., the steady state depends on the initialization. A small subset of the nodes called leaders, which have global objectives and independent dynamics than the others, are typically used to guarantee a unique equilibrium. We also adopt leader(s) shown in Fig. \ref{fig: used graphs} for a fair error comparison. Note that if agent $i \in \mathcal{V}$ is a leader, then the nature of $\bm{H}_i$ does not affect the system due to their different dynamics.

\begin{figure}[t]
	\centering	

        \subfloat[]{\raisebox{0ex}
		{\includegraphics[width=0.11\textwidth]{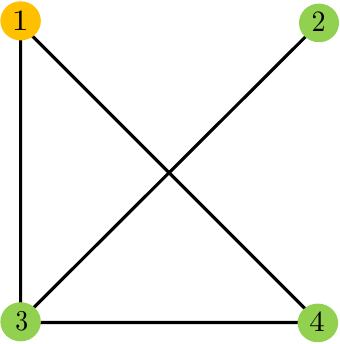}}%
	}
	\hspace{2ex}
	\subfloat[]{\raisebox{0ex}
		{\includegraphics[width=0.17\textwidth]{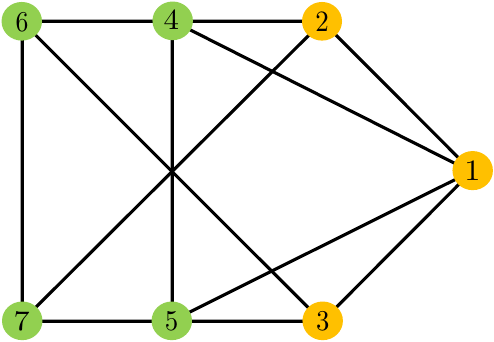}}%
	}

	\caption{The graphs for (a) rendezvous control and (b) distributed formation control, where the orange nodes are leaders.}
	\label{fig: used graphs}
		
\end{figure}

\subsection{Case 1: Rendezvous Control}
The rendezvous control algorithm \cite{amirkhani2022consensus,lin2003multi}, which originates from the classic average consensus algorithm \cite{olfati2007consensus}, ensures all agents converge to a common location. It involves a standard graph Laplacian $\bm{L}$ that has rank $N-1$. For graph Fig. \ref{fig: used graphs} (a), there are $P=3$ non-zero eigenvalues for the Laplacian. We set node $1$ to be the leader with a constant value $\bm{z}_1=[0,0]^T$. Then the equilibrium is zero, i.e., $\bm{z}_e = [0,0]^T$. We define the error as $\delta(t) = \norm{\bm{z}(t)-\bm{z}_e} = \norm{\bm{z}(t)}$.



\begin{figure}[t]
    \centering%
    
    \subfloat{
        \includegraphics[width=0.8\linewidth]{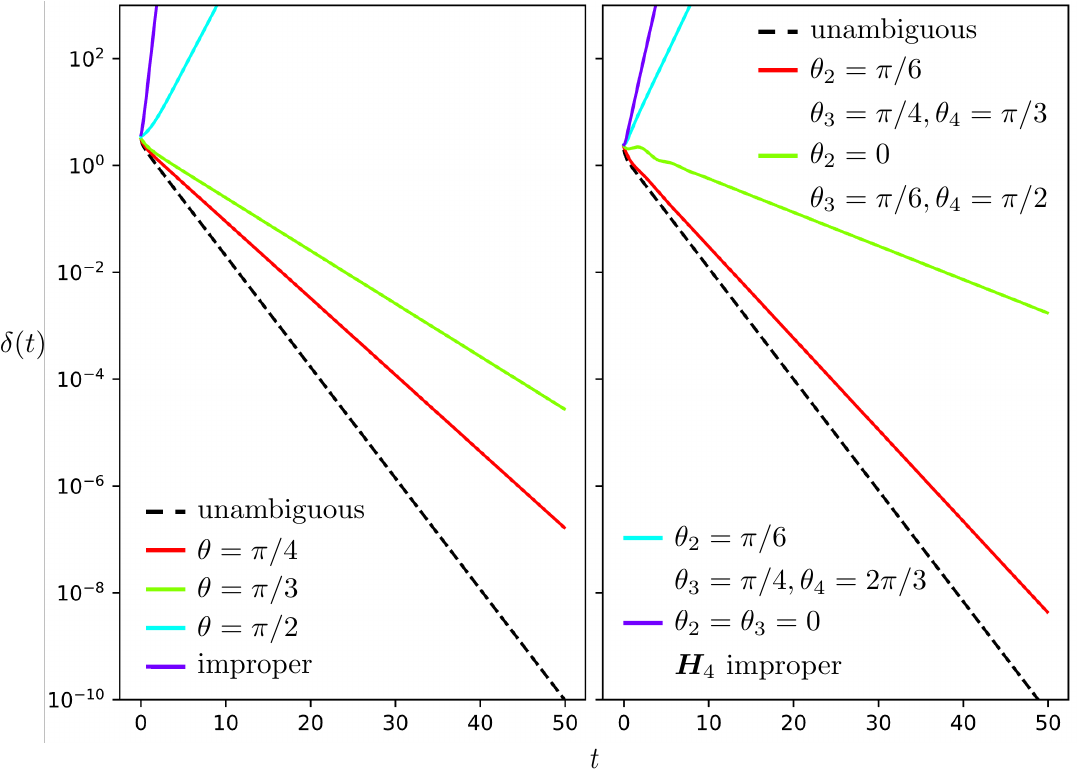}%
        \label{fig: sim avg con}
    }%
    
    \vspace{0mm}  
    
    \subfloat{
        \includegraphics[width=0.8\linewidth]{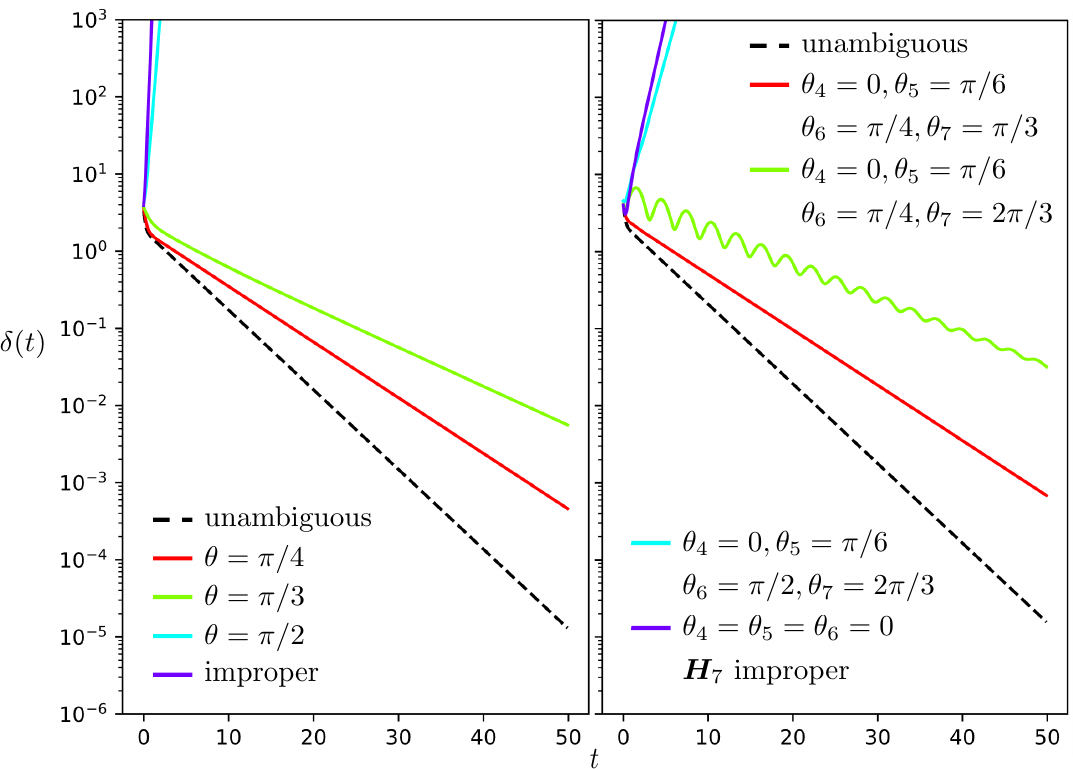}%
        \label{fig: sim fc}
    }%
    
    \caption{The convergence in error $\delta(t)$ across time $t$ for the rendezvous control (top) and the affine formation control algorithm (bottom) under homogeneous ambiguities (left) and heterogeneous (right).}%
    \label{fig: sim cases}%
\end{figure}

\subsection{Case 2: Distributed Formation Control}
Affine formation control \cite{lin2015necessary,zhao2018affine,li2023geometry} is a type of distributed formation control method that can also fit under the consensus framework. A generalized Laplacian, called a stress matrix \cite{alfakih2011bar} with $P=N-D-1$ non-zero eigenvalues, is then used instead of a standard Laplacian. Here, the desired formation is considered the equilibrium point of the system. We consider Fig. \ref{fig: used graphs} (b) in $\mathbb{R}^2$ with a equilibrium $\bm{z}_e=\qty[2,0,1,1,1,-1,0,1,0,-1-1,1,-1,-1]^T$. If we define the first three agents as leaders that remain at their respective target positions $\qty[\bm{z}_1^T,\bm{z}_2^T,\bm{z}_3^T] = \qty[2,0,1,1,1,-1]^T$, then the agents will converge to the defined equilibrium $\bm{z}\rightarrow\bm{z}_e$ as time $t\rightarrow\infty$ given any random initialization of the follower' positions. As such, we define the error $\delta(t) = \norm{\bm{z}(t)-\bm{z}_e}$.

\subsection{Discussion} The numerical results for both rendezvous control and affine formation control are shown in Fig. \ref{fig: sim cases}, where the cases discussed in Section \ref{sec: general conditions} are simulated as compared to unambiguous cases. As can be seen, the errors present an exponential decay (strain lines under log scale) if converging. In the homogeneous case, errors are converging for $\theta\in(-\frac{\pi}{2},\frac{\pi}{2})$ cases and diverging for $\theta=\frac{\pi}{2}$ and improper rotations, which agree to Theorem \ref{thm: homo proper rot} and \ref{thm: homo improp rot}. In the heterogeneous case, errors are converging in cases where $\theta_i\in(-\frac{\pi}{2},\frac{\pi}{2})$ for all followers and where one exceeds this range, which proves the sufficient but not necessary condition in Theorem \ref{thm: hetero proper rot}. We observe that the error diverges with one follower under improper rotations even if the other agents are unambiguous, which supports Proposition \ref{prop: hetero mix rot}.

\section{Conclusion}\label{sec: conclusions}
In this work, we conducted a theoretical analysis of the stability of consensus control where the local reference frames are subject to rotational ambiguities. We show that the system is robust to proper rotations in both homogeneous and heterogeneous cases within certain margins, however the stability is compromised in the presence of improper rotations. This provides insightful guidance for the design of relative localization and the implementation of consensus control in various applications given that state-of-art solutions assume an aligned local reference frames for all nodes. We observe in numerical examples that different rotation angles have a varying impact on the convergence rate. In our ongoing work, we aim to study this relationship in depth and explore the effect of ambiguities on consensus systems at large.

\appendices

\section{Alternative Proof for Stability under Homogeneous Rotations}\label{apdx: alt proofs}

\begin{lemmax}\label{lmm: kron eigs}
        (\textit{Eigenvalues and eigenvectors of Kronecker products of matrices} \cite{broxson2006kronecker}) 
            Suppose $\bm{A}$ and $\bm{B}$ are square matrices of size $N$ and $M$ respectively and they admit $\bm{A}\bm{v}_n=\lambda_n\bm{v}_n$ for $n=1,...,N$ and $\bm{B}\bm{w}_m=\mu_m\bm{w}_m$ for $m=1,...,M$, then $\bm{v}_n\otimes\bm{w}_m$ is an eigenvector of $\bm{A}\otimes\bm{B}$ corresponding to the eigenvalue $\lambda_n\mu_m$. Additionally, the set of all eigenvalues of $\bm{A}\otimes\bm{B}$ is $\qty{\lambda_n\mu_m:n=1,...,N, m=1,...,M}$.
        \end{lemmax}

\begin{lemmax}\label{lmm: alt proper rot}
(\textit{Rotation of eigenvalues of proper rotations})
The eigenvalues of the rotated system (\ref{equ: homo rot}) are the ones of the negative Laplacian $-\bm{L}$ rotated in the complex plane by $\qty{\theta, -\theta}$ if $D=2$ and $\qty{0,\theta,-\theta}$ if $D=3$, given $\bm{H} = \mathcal{R}(\theta)$.
\end{lemmax}
\begin{proof}
Observe $-(\bm{I}_N\otimes \bm{H})(\bm{L}\otimes\bm{I}_D)= -(\bm{L}\otimes\bm{H})$. Let $\lambda_n$ for $n=1,...,N$ denote the eigenvalues of $-\bm{L}$ and $\mu_d$ for $d=1,...,D$ denote those of $\bm{H}$. Based on Lemma \ref{lmm: kron eigs}, the eigenvalues of $-(\bm{L}\otimes\bm{H})$ are $\mu_d\lambda_n$ for $d=1,...D, n=1,...,N$. Since $\sigma\qty(\bm{H}) = \{e^{ i\theta}, e^{ -i\theta}\}$ and $\sigma\qty(\bm{H}) = \{1,e^{ i\theta},e^{-i\theta}\}$ for $D=2$ and $3$, respectively, the resulting eigenvalues are $\qty{e^{i\theta}\lambda_n, e^{-i\theta}\lambda_n}$ and $\qty{\lambda_n, e^{i\theta}\lambda_n, e^{-i\theta}\lambda_n}$ for $n=1,...,N$, for respective dimensions $D$. Thus, the eigenvalues are rotated in the complex plane by $\theta$ and $-\theta$ for $D=2$ and $\theta$, $-\theta$ and $0$ for $D=3$.
\end{proof}

\begin{lemmax}\label{lmm: alt improper rot}
(\textit{Rotation and mirroring of eigenvalues of improper rotations})
Let $\bm{H}$ be an improper rotation matrix, and $\bm{L}$ be a Laplacian (\ref{equ: evd Lap}), then there always exists positive eigenvalues for $-(\bm{I}_N\otimes \bm{H})(\bm{L}\otimes\bm{I}_D)$.
\end{lemmax}

\begin{proof}   
    We simplify the again $-(\bm{I}_N\otimes \bm{H})(\bm{L}\otimes\bm{I}_D) = -(\bm{L}\otimes\bm{H})$, whose eigenvalues are $\mu_d\lambda_n$ for $d=1,...D, n=1,...,N$. It is known that $\sigma\qty(\bm{H}) = \{-1, 1\}$ for $D=2$ and $\sigma\qty(\bm{H}) =\qty{-1, e^{i\theta}, e^{-i\theta}}$ for $D=3$ for an improper $\bm{H}$. Hence, there exists a set of eigenvalues of $-\bm{L}$ mirrored from the negative part to the positive part of the real axis by the $-1$ eigenvalue of $\bm{H}$.
\end{proof}

\bibliographystyle{IEEEtran}
\bibliography{refs}

\end{document}